\documentclass[a4paper]{svproc}
\usepackage{url}

\usepackage{amsmath,amssymb}
\newtheorem{assumption}{Assumption}
\usepackage{tikz}
\usetikzlibrary{arrows.meta,positioning,calc,fit,shapes.geometric}
\usepackage{graphicx}
\usepackage{subcaption}
\usepackage[english]{babel}
\begin{document}
\mainmatter              
\title{High-Probability ISS Tubes for Continuous-Time State Estimation}
\titlerunning{High-Probability ISS Tubes}
\author{Jerzy Baranowski}
\authorrunning{Jerzy Baranowski} 
%
\tocauthor{Jerzy Baranowski}
\institute{AGH University of Kraków,\\
\email{jb@agh.edu.pl}}

\maketitle              

\begin{abstract}
This paper studies a probabilistic interpretation of input-to-state
stability (ISS) bounds for estimation-error dynamics in continuous-time systems.
We show that, if the aggregated disturbance satisfies a probabilistic envelope in
an essential-supremum sense, then deterministic ISS bounds immediately induce
high-probability error tubes. To make this interpretation constructive, we also
provide explicit sufficient conditions based on quadratic Lyapunov inequalities
and specialize them to positive and cooperative systems. The approach is illustrated
on a positive compartment model with aggregated measurements, where ISS tubes
are compared with nominal uncertainty bands produced by a Kalman--Bucy filter
and by Gaussian and robust moving-horizon estimators. The examples show that
ISS tubes provide a conservative but computationally light uncertainty baseline,
while robust MHE is less sensitive to outlier contamination than Gaussian-based
estimators.
\end{abstract}

\section{Introduction}

Robust state estimation for dynamical systems is central in control and monitoring
applications. A classical robustness approach treats disturbances and noise as bounded
signals and derives worst-case guarantees for the estimation error, while probabilistic
methods quantify uncertainty by specifying stochastic noise models and producing
credible sets or covariance-based bands.

A standard robustness framework for such questions is input-to-state stability
(ISS). Foundational ISS results characterize how the state, or in the present
context the estimation error, depends on the initial condition and on exogenous
disturbance inputs \cite{sontag1995input,sontag1996characterizations,sontag2008iss}. Related extensions cover broader input-output
stability formulations and impulsive settings \cite{hespanha2008lyapunov,krichman2001iooss}. In estimation problems,
this viewpoint naturally leads to deterministic error tubes that remain valid
under bounded disturbances and model mismatch.

This robustness perspective is especially relevant for constrained and positive
systems. Positive-systems theory provides a natural setting for state variables
that represent nonnegative masses, populations, or concentrations \cite{kaczorek2002positive,farina2000positive,kaczorek2008fractional},
while the control-theoretic background used here follows standard ISS-based
reasoning \cite{mitkowski2019zarys}. In such settings, a deterministic tube description may be
more credible than a fully specified stochastic model, especially when the
available knowledge is better expressed through envelopes or bounds.

In parallel, probabilistic estimation methods describe uncertainty through
distributional assumptions. Classical Kalman filtering provides covariance-based
uncertainty quantification for linear-Gaussian models \cite{kalman1960new}, while moving-
horizon estimation offers an optimization-based alternative that has been
extensively studied for constrained and nonlinear systems \cite{rao2003constrained,haseltine2005critical}. More
recent work has developed robust, variational-Bayes, and data-driven variants
of MHE and related Bayesian approximations \cite{liu2013robust,dong2022variational,sun2023data,fang2018nonlinear,battistelli2018distributed}.

The present paper does not propose a new estimator. Instead, it provides an
interpretation layer that connects deterministic ISS robustness analysis with
probabilistic uncertainty statements. More precisely, we show that if the
disturbance satisfies a probabilistic envelope in an essential-supremum sense,
then the deterministic ISS tube can be read as a high-probability credibility
region. The main result is formulated for continuous-time systems, linked to
explicit sufficient conditions based on quadratic Lyapunov inequalities, and
illustrated on a positive compartment model with aggregated measurements under
both matched bounded noise and outlier contamination.

\section{Problem formulation, estimation error dynamics and ISS tubes}

Consider a continuous-time system
\begin{equation}
\dot x(t) = f(x(t),u(t),t) + w(t), \qquad
y(t) = h(x(t),t) + v(t),
\end{equation}
where $x(t)\in\mathbb{R}^n$ is the state, $y(t)\in\mathbb{R}^m$ is the measured output,
$w(t)$ denotes process disturbances, and $v(t)$ denotes measurement noise.

Let $\hat x(t)$ be a state estimate generated by an observer or estimation algorithm.
Let $e(t) := x(t) - \hat x(t).\in\mathbb{R}^n$ denote the estimation error and assume its dynamics can be
written in the abstract form
\begin{equation}
\dot e(t) = f_e(e(t),d(t),t),
\label{eq:abstract-error}
\end{equation}
where $d(t)$ is an aggregated disturbance/noise input. For a measurable signal
$s:[0,t]\to\mathbb{R}^q$ we define the essential supremum norm
\[
\|s\|_{\infty,[0,t]} := \operatorname*{ess\,sup}_{\tau\in[0,t]}\|s(\tau)\|.
\]

\begin{assumption}[Deterministic ISS bound]
\label{ass:iss}
There exist functions $\beta\in\mathcal{KL}$ and $\gamma\in\mathcal{K}$ such that for
all $t\ge 0$,
\begin{equation}
\|e(t)\| \le \beta(\|e(0)\|,t) + \gamma(\|d\|_{\infty,[0,t]}).
\label{eq:iss-bound}
\end{equation}
\end{assumption}


\begin{assumption}[Probabilistic disturbance envelope]
\label{ass:envelope}
Fix $t\ge 0$. There exist $\bar d(t)\ge 0$ and $\delta\in(0,1)$ such that
\begin{equation}
\mathbb{P}\big(\|d\|_{\infty,[0,t]} \le \bar d(t)\big)\ge 1-\delta .
\label{eq:prob-envelope}
\end{equation}
\end{assumption}

\begin{theorem}[High-probability ISS tube]
\label{thm:hp-tube}
Let $t\ge 0$ be fixed. If Assumptions~\ref{ass:iss} and \ref{ass:envelope} hold, then
\begin{equation}
\mathbb{P}\Big(
\|e(t)\| \le \beta(\|e(0)\|,t) + \gamma(\bar d(t))
\Big)\ge 1-\delta .
\label{eq:hp-tube}
\end{equation}
\end{theorem}

\begin{proof}
Define the event
\[
E_t := \Big\{\omega:\ \|d\|_{\infty,[0,t]}(\omega)\le \bar d(t)\Big\}.
\]
By Assumption~\ref{ass:envelope}, $\mathbb{P}(E_t)\ge 1-\delta$.

Consider any outcome $\omega\in E_t$. Then $\|d\|_{\infty,[0,t]}(\omega)\le \bar d(t)$.
Applying the deterministic ISS inequality \eqref{eq:iss-bound} (Assumption~\ref{ass:iss})
to this outcome yields
\[
\|e(t)\|(\omega)\le \beta(\|e(0)\|,t) + \gamma(\|d\|_{\infty,[0,t]}(\omega))
\le \beta(\|e(0)\|,t) + \gamma(\bar d(t)).
\]
Hence $E_t$ is a subset of the event
\[
T_t := \Big\{\omega:\ \|e(t)\|(\omega)\le \beta(\|e(0)\|,t) + \gamma(\bar d(t))\Big\}.
\]
Therefore $\mathbb{P}(T_t)\ge \mathbb{P}(E_t)\ge 1-\delta$, which is exactly
\eqref{eq:hp-tube}.
\end{proof}


Theorem~\ref{thm:hp-tube} reduces probabilistic tube guarantees to verifying a
deterministic ISS estimate. We now provide an explicit sufficient condition, stated
in a form that directly matches the quadratic Lyapunov constructions used later in
Section~4.

\paragraph{Aggregated disturbance.}
In what follows, we aggregate process disturbances and measurement noise into a single
input $d(t)$ defined as
\begin{equation}
d(t) := \begin{bmatrix} w(t) \\ v(t) \end{bmatrix},
\end{equation}
equipped with the weighted norm
\begin{equation}
\|d(t)\|^2 := \|w(t)\|^2 + \rho^2 \|v(t)\|^2,
\label{eq:d-weighted-norm}
\end{equation}
where $\rho>0$ is a fixed weighting parameter.

\begin{proposition}[Quadratic Lyapunov inequality implies an ISS bound]
\label{prop:lyap-to-iss}
Assume there exists a symmetric matrix $P\succ 0$ and constants $a>0$, $b>0$ such that
the function $V(e)=e^\top P e$ satisfies
\begin{equation}
\dot V(e(t)) \le -a V(e(t)) + b\|d(t)\|^2
\qquad\text{for almost all } t\ge 0,
\label{eq:lyap-diss}
\end{equation}
where $\|\cdot\|$ is defined in \eqref{eq:d-weighted-norm}. Then Assumption~\ref{ass:iss} holds with the explicit functions
\begin{equation}
\beta(r,t) = \sqrt{\frac{\lambda_{\max}(P)}{\lambda_{\min}(P)}}\,e^{-at/2}\,r,
\qquad
\gamma(s) = \sqrt{\frac{b}{a\,\lambda_{\min}(P)}}\,s.
\label{eq:beta-gamma}
\end{equation}
\end{proposition}

\begin{proof}[sketch]
From \eqref{eq:lyap-diss} we have $\dot V \le -aV + b\|d(t)\|^2$. Solving this scalar
differential inequality yields
$V(t)\le e^{-at}V(0) + \frac{b}{a}(1-e^{-at})\|d\|_{\infty,[0,t]}^2$.
Using $\lambda_{\min}(P)\|e\|^2 \le V(e)\le \lambda_{\max}(P)\|e\|^2$ and
$\sqrt{\alpha+\beta}\le \sqrt{\alpha}+\sqrt{\beta}$ gives \eqref{eq:iss-bound} with
$\beta,\gamma$ as in \eqref{eq:beta-gamma}.
\end{proof}

\section{Sufficient conditions via quadratic Lyapunov functions}

This section provides explicit sufficient conditions under which the deterministic
ISS bound \eqref{eq:iss-bound} holds. The focus is on continuous-time systems and
quadratic Lyapunov functions with constant matrices, enabling transparent analysis
and direct connection to error tubes.

\subsection{General nonlinear systems}

Consider the plant--observer pair
\begin{align}
\dot x(t) &= f(x(t),u(t),t) + w(t), \label{eq:plant}\\
y(t) &= h(x(t),t) + v(t), \label{eq:meas}\\
\dot{\hat x}(t) &= f(\hat x(t),u(t),t) + K\big(y(t)-h(\hat x(t),t)\big), \label{eq:observer}
\end{align}
where $x(t),\hat x(t)\in\mathbb{R}^n$, $y(t)\in\mathbb{R}^m$, and $w(t),v(t)$ are
measurable disturbance /noise signals. Define the estimation error $e(t)=x(t)-\hat x(t)$.
Subtracting \eqref{eq:observer} from \eqref{eq:plant} and using \eqref{eq:meas} yields
\begin{equation}
\dot e(t) = f(x(t),u(t),t)-f(\hat x(t),u(t),t)
- K\big(h(x(t),t)-h(\hat x(t),t)\big)
+ w(t)-Kv(t).
\label{eq:error-dynamics}
\end{equation}

\begin{corollary}[ISS via incremental dissipativity in a constant metric]
\label{cor:incremental-full}
Assume that there exist a symmetric matrix $P\succ 0$, constants $\lambda>0$ and
$L_h\ge 0$ such that for all $x,\hat x\in\mathbb{R}^n$ and all $t$:
\begin{enumerate}
\item[(A1)] (Incremental dissipativity)
\begin{equation}
(x-\hat x)^\top P\big(f(x,u,t)-f(\hat x,u,t)\big)
\le -\lambda \|x-\hat x\|^2 .
\label{eq:A1}
\end{equation}
\item[(A2)] (Lipschitz output map)
\begin{equation}
\|h(x,t)-h(\hat x,t)\| \le L_h \|x-\hat x\| .
\label{eq:A2}
\end{equation}
\item[(A3)] (Gain compatibility)
\begin{equation}
\|PK\|\,L_h \le \lambda/2 ,
\label{eq:A3}
\end{equation}
where $\|\cdot\|$ denotes the induced operator norm consistent with the vector norm
$\|\cdot\|$.
\end{enumerate}
Then the estimation error system \eqref{eq:error-dynamics} is input-to-state stable with
respect to the aggregated disturbance $d(t)=(w(t),v(t))$ in the following explicit sense:
there exist $\beta\in\mathcal{KL}$ and $\gamma\in\mathcal{K}$ such that for all $t\ge 0$,
\[
\|e(t)\| \le \beta(\|e(0)\|,t) + \gamma\big(\|w\|_{\infty,[0,t]}+\|v\|_{\infty,[0,t]}\big).
\]
\end{corollary}
\begin{proof}[Proof sketch]
Use $V(e)=e^\top P e$ and differentiate along \eqref{eq:error-dynamics}.
Assumption (A1) provides a negative drift term in $\|e\|^2$. The output injection term is
bounded by (A2) and controlled by (A3). The disturbance terms are handled by Young's
inequality, yielding an estimate of the form $\dot V \le -aV + b\|d\|^2$ for suitable
constants $a,b>0$ and a weighted norm on $(w,v)$. Proposition~\ref{prop:lyap-to-iss}
then implies the ISS bound \eqref{eq:iss-bound}.
\end{proof}
Here the separated form $(w,v)$ is used only for interpretability; by norm
equivalence it can be absorbed into the aggregated disturbance input $d$ used in
Assumption 1. Corollary~\ref{cor:incremental-full} replaces abstract ISS Lyapunov existence assumptions
by explicit incremental dissipativity and Lipschitz-type conditions. These are checkable
and depend directly on the right-hand side of the nonlinear dynamics.

\subsection{Positive and cooperative systems}

Positive and cooperative systems form an important subclass where constant diagonal
Lyapunov functions often arise naturally, making them well suited for illustrating
the probabilistic ISS-tube interpretation.

\begin{corollary}[Positive/cooperative subclass with diagonal metric]
\label{cor:positive-full}
Assume that the system \eqref{eq:plant} is positive on $\mathbb{R}^n_{\ge 0}$ and cooperative
on a forward-invariant region of interest, and that the conditions (A2)--(A3) of
Corollary~\ref{cor:incremental-full} hold. If, in addition, there exists a \emph{diagonal}
matrix $P\succ 0$ and $\lambda>0$ such that the incremental dissipativity condition (A1)
holds for all $x,\hat x\ge 0$, then the estimation error system is ISS with respect to
$(w,v)$, and the ISS bound \eqref{eq:iss-bound} holds.
\end{corollary}

\begin{proof}
Under the stated assumptions, Corollary~\ref{cor:incremental-full} applies verbatim with
the same Lyapunov function $V(e)=e^\top P e$. The diagonal restriction on $P$ does not
change the proof; it only specializes the metric to a form that is compatible with the
geometry commonly used in positive/cooperative systems (e.g., diagonal stability and
weighted norm constructions). Therefore, the ISS estimate \eqref{eq:iss-bound} holds. Positivity/cooperativity ensures that the considered
region (e.g., $\mathbb{R}^n_{\ge 0}$) is forward invariant and motivates the use of such
diagonal metrics in applications, but the ISS argument itself follows directly from
Corollary~\ref{cor:incremental-full}.
\end{proof}


\section{Simulation study}

This section illustrates the proposed interpretation on a continuous-time positive
compartment model with aggregated measurements. We compare Kalman--Bucy and Gaussian
and robust MHE uncertainty bands with deterministic and high-probability ISS tubes.
All estimators are tested on identical Monte Carlo realizations of $(w,v)$ and the
same sample-and-hold measurement stream. Regime R1 matches the bounded-signal assumptions
behind ISS analysis, while R2 introduces measurement outliers that violate Gaussian
noise assumptions. This setup allows direct comparison of empirical coverage, constraint
violations, and the conservativeness of ISS-based containment.




\subsection{Model and measurements}

We consider a positive compartment system
\begin{equation}
\dot{x}(t)=Ax(t)+w(t), \qquad x(t)\in\mathbb{R}^3_{\ge 0},\label{eq:sim-model}
\end{equation}
with the Metzler, Hurwitz matrix
\begin{equation}
A=
\begin{bmatrix}
-0.65 & 0.15 & 0 \\
0.35 & -0.60 & 0.10 \\
0 & 0.25 & -0.35
\end{bmatrix}.\label{eq:sim-A}
\end{equation}

The example should be understood as a stylized model of a positive dynamical
system rather than as an identified model of one specific plant. Compartment
structures of this type are common when the state variables represent
nonnegative masses, concentrations, inventories, or populations, and when only
aggregate sensing is available. This makes the example representative of a
broad class of practically relevant monitoring problems in which positivity is
structural and full state measurement is not available.

Measurements are aggregated as total mass and one compartment,
\begin{equation}
y(t)=Cx(t)+v(t), \qquad
C=
\begin{bmatrix}
1 & 1 & 1 \\
1 & 0 & 0
\end{bmatrix}.\label{eq:sim-C}
\end{equation}

This sensing structure intentionally creates partial observability: one channel
measures a total balance, while the second measures only a selected component.
Such a setup is useful here because it reflects realistic sensing limitations and
creates a meaningful setting for comparing conservative ISS-based uncertainty
tubes with distribution-dependent Bayesian uncertainty bands.

\subsection{Disturbance and noise regimes}





We consider two regimes.

\paragraph{R1: bounded disturbances (matched).}
The signals $w(t)$ and $v(t)$ are bounded and piecewise constant:
\begin{equation}
|w_i(t)|\le \bar w,\quad |v_j(t)|\le \bar v,\qquad
\bar w=0.02,\ \bar v=0.01.
\label{eq:R1-bounds}
\end{equation}

\paragraph{R2: outlier-contaminated measurements (mismatch).}
The process disturbance remains as in R1. Measurement noise at sampling instants
$t_k=k\Delta t_y$ follows a mixture model: with probability $1-p$, $v_k\sim\mathcal{N}(0,\sigma^2 I)$,
and with probability $p$, $v_k\sim\mathcal{N}(0,(\kappa\sigma)^2 I)$. We use
\begin{equation}
\sigma=0.005,\quad p=0.03,\quad \kappa=15.
\label{eq:R2-mix}
\end{equation}

Measurements are sampled and held constant. Denoting $K=\lfloor t/\Delta t_y\rfloor$,
the essential supremum reduces to a maximum over samples:
\begin{equation}
\|v\|_{\infty,[0,t]} = \max_{0\le k\le K}\|v(t_k)\|.
\label{eq:samplehold-sup}
\end{equation}







\subsection{Estimators}

We compare three estimators.

\paragraph{E1: Kalman--Bucy filter (Gaussian).}
A continuous-time Kalman--Bucy filter is implemented with covariances
\begin{equation}
Q=qI_3,\quad q=10^{-4},\qquad R=rI_2,\quad r=2.5\cdot 10^{-5}.
\label{eq:QR}
\end{equation}
Credibility intervals are extracted from the covariance $P(t)$ as marginal intervals.

\paragraph{E2: positivity-aware MAP/MHE (Gaussian).}
A moving-horizon MAP estimator is defined from Gaussian process and measurement penalties
and nonnegativity constraints on states. Credible intervals are obtained via a Laplace
approximation using the local Hessian at the MAP solution.

\paragraph{E3: robust Bayesian MHE.}
A robust MHE variant replaces the Gaussian measurement likelihood by a heavy-tailed
alternative, aiming to mitigate outliers in regime R2.

\subsection{ISS tube via quadratic Lyapunov function}

To compute a deterministic reference tube, we use a constant-gain observer
\begin{equation}
\dot{\hat x}(t) = A\hat x(t) + L\big(y(t)-C\hat x(t)\big),
\label{eq:ISS-observer}
\end{equation}
with
\begin{equation}
L=
\begin{bmatrix}
1.20 & 0.80\\
0.60 & 0.10\\
0.90 & 0.00
\end{bmatrix},
\qquad F:=A-LC.
\label{eq:ISS-gain}
\end{equation}
The error $e(t)=x(t)-\hat x(t)$ satisfies
\begin{equation}
\dot e(t) = F e(t) + d(t),\qquad d(t)=w(t)-Lv(t).
\label{eq:error-linear}
\end{equation}

Let $P\succ 0$ solve the Lyapunov equation
\begin{equation}
F^\top P + PF = -I.
\label{eq:lyap-eq}
\end{equation}
Then $V(e)=e^\top P e$ satisfies the dissipation inequality
\begin{equation}
\dot V(e(t)) \le -aV(e(t)) + b\|d(t)\|^2,
\label{eq:Vdot-ab}
\end{equation}
for explicit constants $a>0$, $b>0$ (cf. Proposition~\ref{prop:lyap-to-iss}).
Consequently,
\begin{equation}
V(t) \le e^{-at}V(0) + \frac{b}{a}\big(1-e^{-at}\big)\|d\|_{\infty,[0,t]}^2,
\label{eq:V-bound}
\end{equation}
and using $\lambda_{\min}(P)\|e\|^2\le V(e)$ we obtain the ISS tube radius
\begin{equation}
r_{\mathrm{ISS}}(t)=
\sqrt{\frac{1}{\lambda_{\min}(P)}\left(
e^{-at}V(0)+\frac{b}{a}(1-e^{-at})\|d\|_{\infty,[0,t]}^2
\right)}.
\label{eq:rISS}
\end{equation}
In regime R1 we use the deterministic bound $\|d(t)\|\le \bar d$ derived from \eqref{eq:R1-bounds}.
In regime R2, $\|d\|_{\infty,[0,t]}$ is bounded with high probability using the lemma below.

\subsection{High-probability envelope in regime R2}

\begin{lemma}[Conservative envelope for sample-and-hold mixture noise]
\label{lem:R2-envelope}
Fix $t\ge 0$ and let $K=\lfloor t/\Delta t_y\rfloor$. For any $\delta\in(0,1)$ define
\begin{equation}
b(t,\delta) := \kappa\sigma\sqrt{2\log\!\left(\frac{8(K+1)}{\delta}\right)}.
\label{eq:btd}
\end{equation}
Then
\begin{equation}
\mathbb{P}\!\left(\max_{0\le k\le K}\|v(t_k)\|_\infty \le b(t,\delta)\right)\ge 1-\delta.
\label{eq:R2-envelope-v}
\end{equation}
where $\Vert \cdot\Vert_\infty$ denotes componentwise maximum norm on $\mathbb{R}^n$
Consequently, with $d(t)=w(t)-Lv(t)$,
\begin{equation}
\mathbb{P}\big(\|d\|_{\infty,[0,t]} \le \bar d(t,\delta)\big)\ge 1-\delta,
\qquad
\bar d(t,\delta):=\sqrt{3}\bar w + \|L\|\sqrt{2}\,b(t,\delta).
\label{eq:R2-envelope-d}
\end{equation}
\end{lemma}

\begin{proof}
For $Z\sim\mathcal{N}(0,s^2)$, $\mathbb{P}(|Z|>b)\le 2\exp(-b^2/(2s^2))$.
For the mixture model \eqref{eq:R2-mix}, the tail is dominated by the larger variance,
hence $\mathbb{P}(|v_j(t_k)|>b)\le 2\exp(-b^2/(2(\kappa\sigma)^2))$.
A union bound over $m=2$ channels and $K+1$ sampling instants yields
\[
\mathbb{P}\big(\exists\,k,j:\ |v_j(t_k)|>b\big)
\le 2m(K+1)\exp\!\left(-\frac{b^2}{2(\kappa\sigma)^2}\right).
\]
Choosing $b=b(t,\delta)$ from \eqref{eq:btd} makes the right-hand side at most $\delta$,
proving \eqref{eq:R2-envelope-v}. The bound \eqref{eq:R2-envelope-d} follows from
$\|v\|_2\le \sqrt{2}\|v\|_\infty$, $\|w\|_2\le \sqrt{3}\bar w$, and
$\|w-Lv\|\le \|w\|+\|L\|\,\|v\|$.
\end{proof}

\begin{figure}[!hb]
\centering

\begin{subfigure}[t]{0.49\textwidth}
  \centering
  \includegraphics[width=\linewidth]{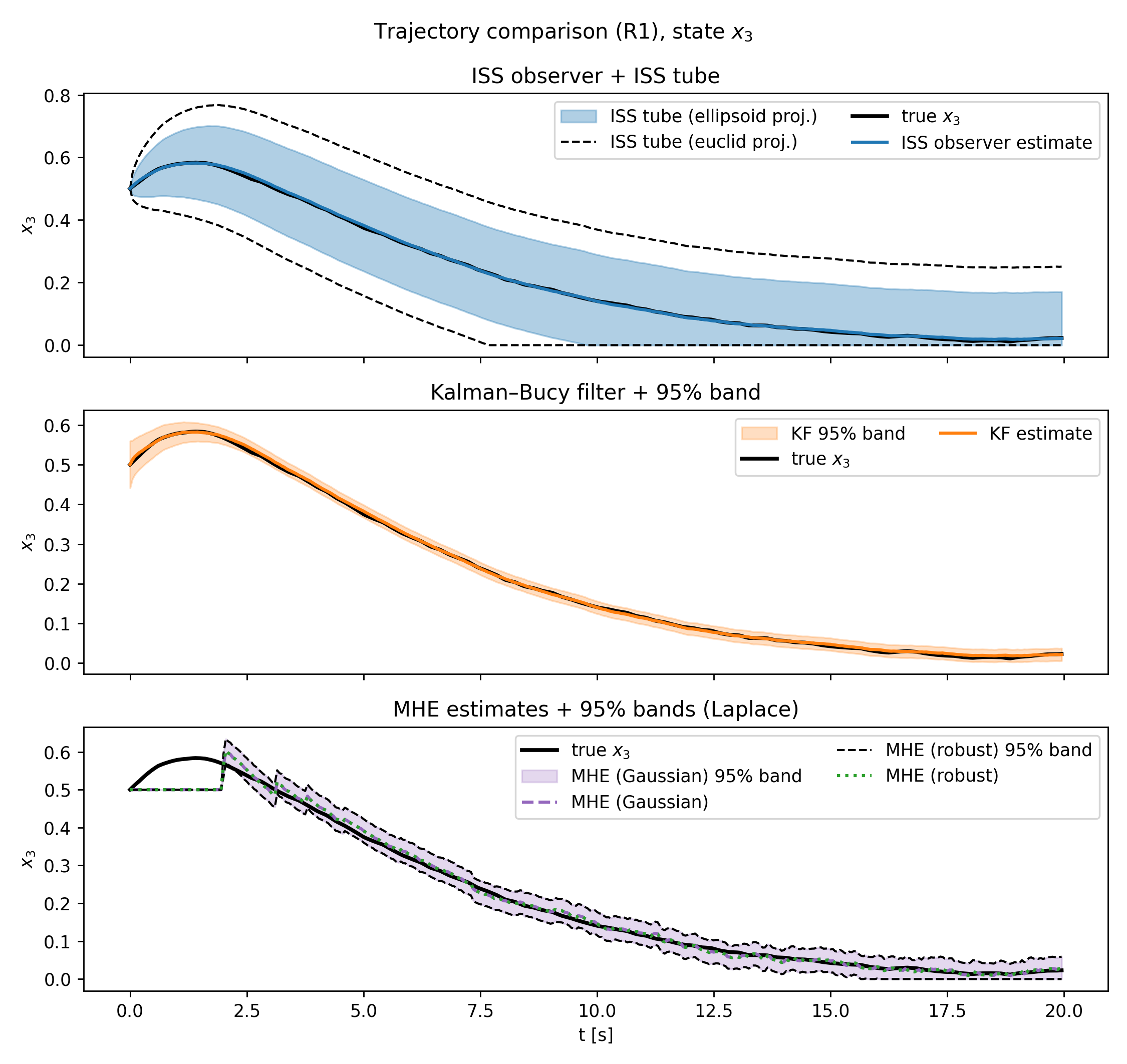}
  \caption{R1 (bounded noise)}
  \label{fig:trace_x3_R1}
\end{subfigure}
\hfill
\begin{subfigure}[t]{0.49\textwidth}
  \centering
  \includegraphics[width=\linewidth]{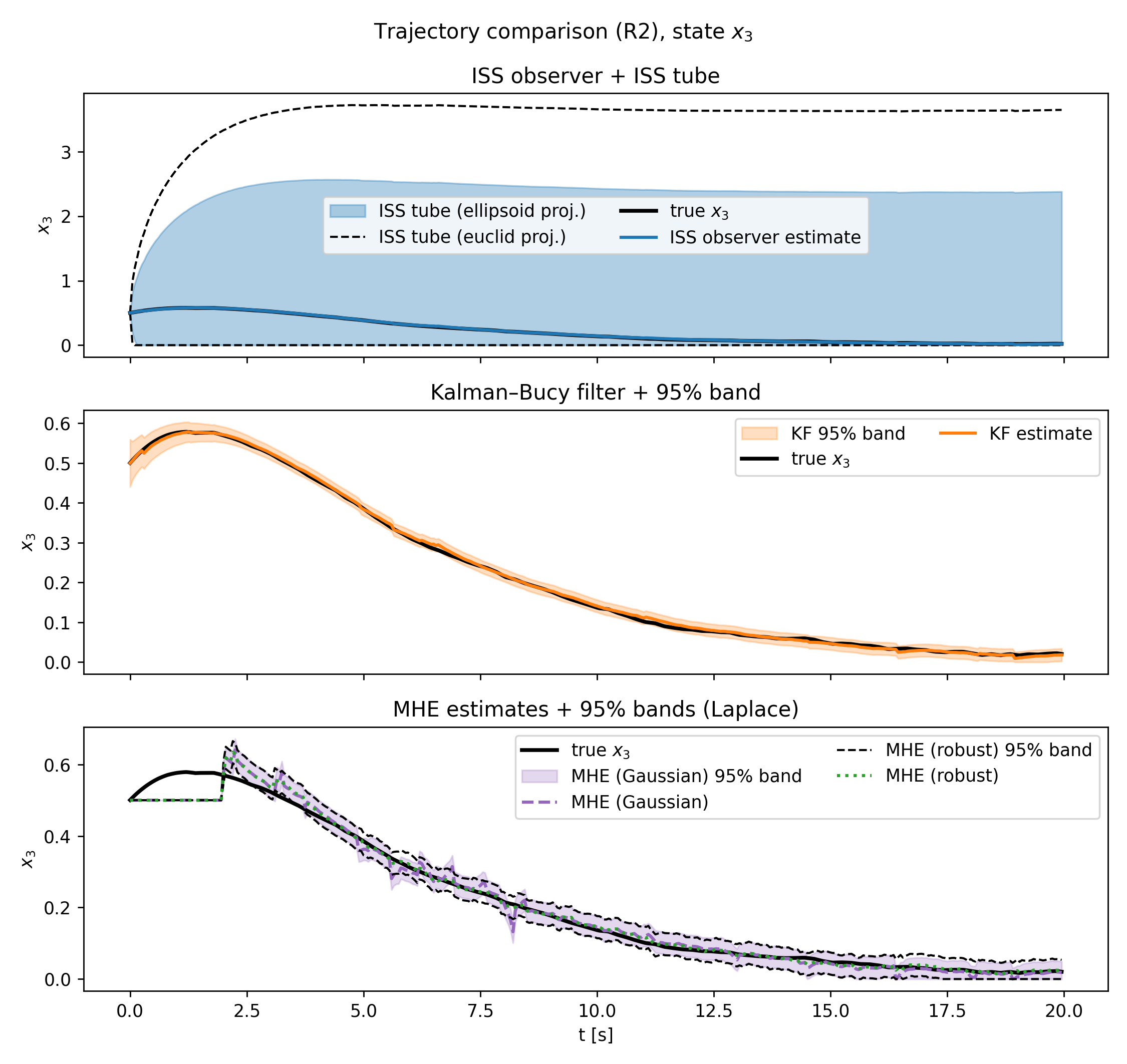}
  \caption{R2 (outlier mixture)}
  \label{fig:trace_x3_R2}
\end{subfigure}

\caption{Representative trajectories and uncertainty descriptions for the most informative state component $x_3$ under regimes R1 and R2. The plots compare the ISS tube, the Kalman--Bucy filter, and Gaussian and robust MHE bands. In R2, outlier contamination makes the Gaussian-based bands less reliable, while the ISS tube acts as a conservative containment reference.}
\label{fig:trace_x3_both}
\end{figure}

\section{Simulation and results}

Table~\ref{tab:sim-params} summarizes the simulation parameters. We report: (i) empirical coverage of nominal credibility intervals and ISS tubes,
(ii) average tube width, (iii) frequency of negative state estimates (constraint
violations), and (iv) qualitative comparison of tube shapes and conservativeness.

\begin{table}
\caption{Simulation parameters.}
\label{tab:sim-params}
\centering
\begin{tabular}{ll}
\hline
Time horizon & $T=20$ s \\
Measurement sampling & $\Delta t_y=0.05$ s (sample-and-hold) \\
Disturbance refresh & $\Delta t_w=0.2$ s (piecewise constant) \\
Monte Carlo runs & $N$ (reported in experiments) \\
R1 bounds & $\bar w=0.02$, $\bar v=0.01$ \\
R2 mixture & $\sigma=0.005$, $p=0.03$, $\kappa=15$ \\
Kalman--Bucy covariances & $Q=10^{-4}I_3$, $R=2.5\cdot 10^{-5}I_2$ \\
ISS observer gain & $L$ given by \eqref{eq:ISS-gain} \\
ISS Lyapunov equation & $(A-LC)^\top P+P(A-LC)=-I$ \\
\hline
\end{tabular}
\end{table}


From a computational viewpoint, the ISS tube construction is intentionally
lightweight. Offline, it requires choosing an observer gain and solving the
Lyapunov equation \eqref{eq:lyap-eq}. Online, it only propagates the observer state and
updates the tube radius through the closed-form bound \eqref{eq:rISS}, together with
the disturbance envelope used in the considered regime. By contrast, the
Kalman--Bucy filter additionally propagates a covariance matrix, while the
Gaussian and robust MHE variants require repeated horizon-based optimization and,
for uncertainty quantification, a local curvature or Hessian-based approximation.
Accordingly, the ISS tube should be interpreted as a computationally inexpensive
conservative robustness baseline rather than as a replacement for posterior
inference. Its advantage is not tighter uncertainty quantification, but a simple
and transparent containment guarantee under disturbance-envelope assumptions.


Figure~\ref{fig:trace_x3_both} shows representative trajectories and uncertainty descriptions for \(x_3\) under both regimes. In regime R1, all estimators track the state well and the nominal 95\%
bands are broadly consistent with the data; differences are mainly visible in band
tightness and in how positivity is handled by constrained MHE. In regime R2, occasional
outliers lead to visible excursions in E1 and in Gaussian MHE bands, which is reflected
in reduced empirical coverage and increased variability, while robust MHE mitigates the
impact of contaminated measurements.

Figure~\ref{fig:coverage_R1_R2} reports empirical coverage of nominal 95\% marginal
bands. Values below $0.95$ indicate under-coverage (miscalibration), while values above
$0.95$ indicate conservativeness. As expected, under regime R2 the Gaussian-based bands
exhibit more pronounced under-coverage due to outliers, whereas robust MHE is closer to
the nominal level. The ISS curve is interpreted as containment frequency in a robustness
tube rather than a calibrated 95\% Bayesian statement.

Figure~\ref{fig:negativity_R1_R2} shows positivity violations. Negativity events are
rare in R1 but become more frequent under R2 for unconstrained estimators due to larger
measurement perturbations. Constrained MHE reduces negativity by construction, providing
a direct diagnostic of whether the positivity prior is respected in the estimator output.

\begin{figure}[!hb]
\centering

\begin{subfigure}[t]{0.49\textwidth}
  \centering
  \includegraphics[width=\linewidth]{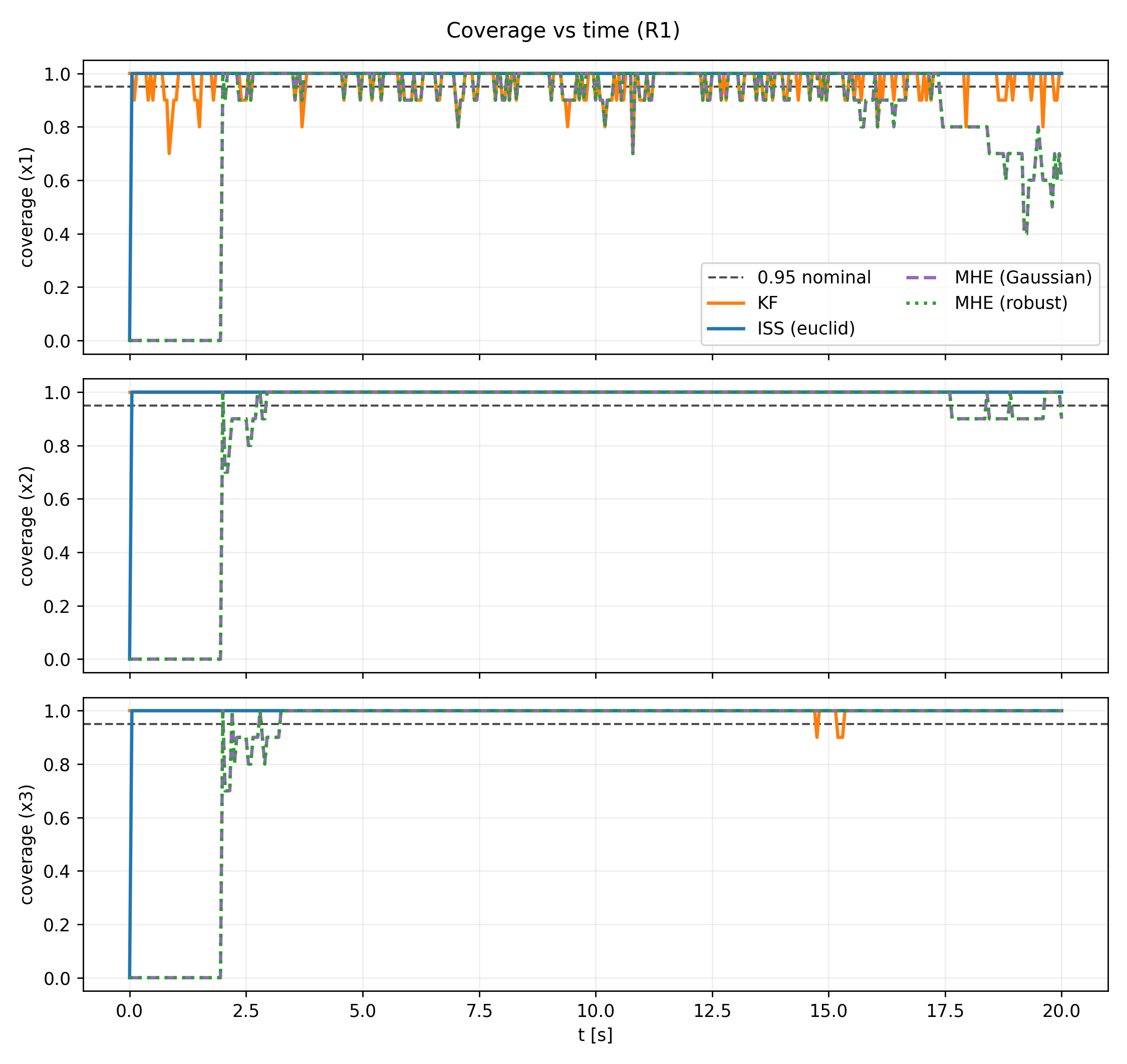}
  \caption{Regime R1 (bounded noise)}
  \label{fig:cov_R1}
\end{subfigure}
\hfill
\begin{subfigure}[t]{0.49\textwidth}
  \centering
  \includegraphics[width=\linewidth]{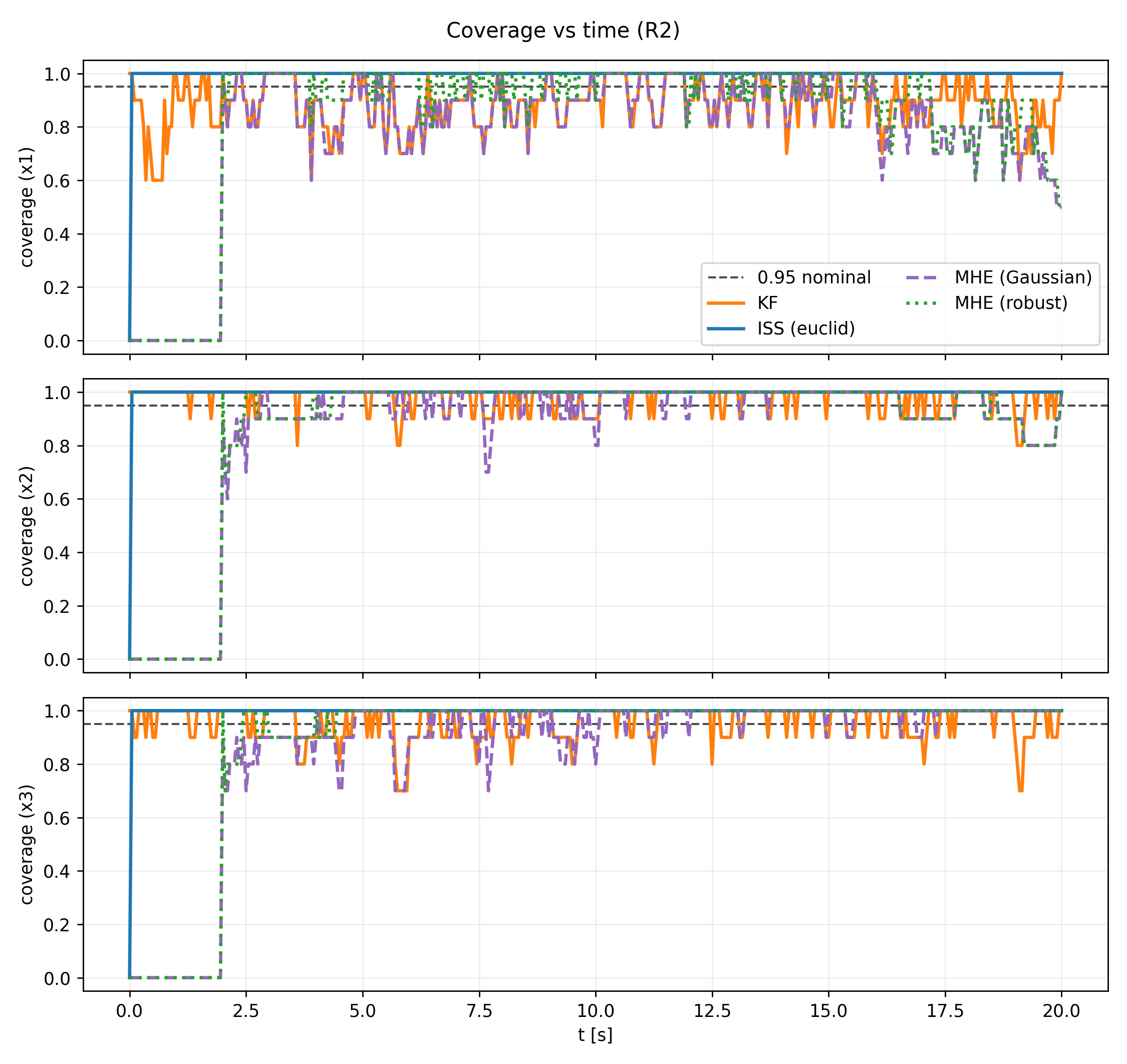}
  \caption{Regime R2 (outlier mixture)}
  \label{fig:cov_R2}
\end{subfigure}

\caption{Empirical coverage of nominal 95\% state bands over time under regimes R1 and R2. For each state component $(x_1,x_2,x_3)$, we report the Monte Carlo frequency that the true state lies inside the estimator's marginal 95\% interval. Values below 0.95 indicate under-coverage, while the ISS curve should be read as robustness-tube containment.}
\label{fig:coverage_R1_R2}
\end{figure}
\section{Conclusions and outlook}

This paper proposed a probabilistic interpretation of ISS-based estimation-error
bounds for continuous-time systems. The main contribution is not a new estimator,
but a bridge between two uncertainty languages: deterministic robustness bounds
from ISS analysis and high-probability credibility regions induced by a
probabilistic disturbance envelope. In this sense, a deterministic ISS tube can
be read as a conservative uncertainty set even when a fully specified stochastic
state-space model is unavailable.

The paper also provided explicit sufficient conditions based on quadratic
Lyapunov inequalities and showed how these conditions specialize to positive and
cooperative systems. In the simulation study, the resulting ISS tubes behaved as
expected: under matched bounded disturbances they provided conservative
containment, while under outlier-contaminated measurements they remained
interpretable as robustness envelopes, in contrast to nominal Gaussian bands,
which exhibited visible under-coverage. The robust MHE variant was less
sensitive to such contamination than Gaussian-based estimators, but at a higher
computational cost.

These results suggest that ISS tubes can serve as a useful baseline for
un-certainty-aware state estimation when robustness assumptions are more credible
than detailed distributional modelling. Future work will address discrete-time
formulations, state-dependent Lyapunov metrics, and extensions to
infinite-dimensional systems.
\begin{figure}[!ht]
\centering

\begin{subfigure}[t]{0.49\textwidth}
  \centering
  \includegraphics[width=\linewidth]{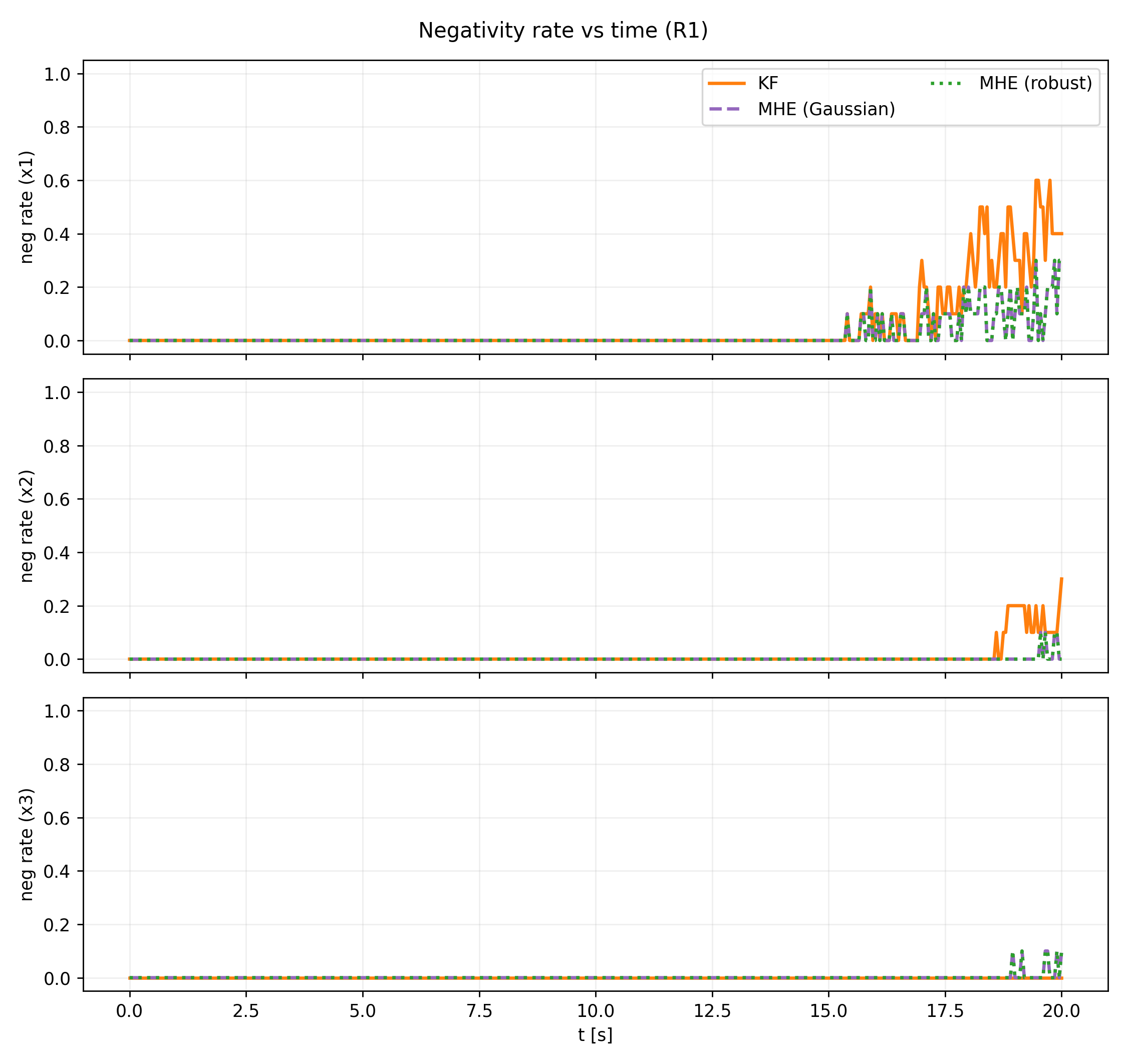}
  \caption{Regime R1 (bounded noise)}
  \label{fig:neg_R1}
\end{subfigure}
\hfill
\begin{subfigure}[t]{0.49\textwidth}
  \centering
  \includegraphics[width=\linewidth]{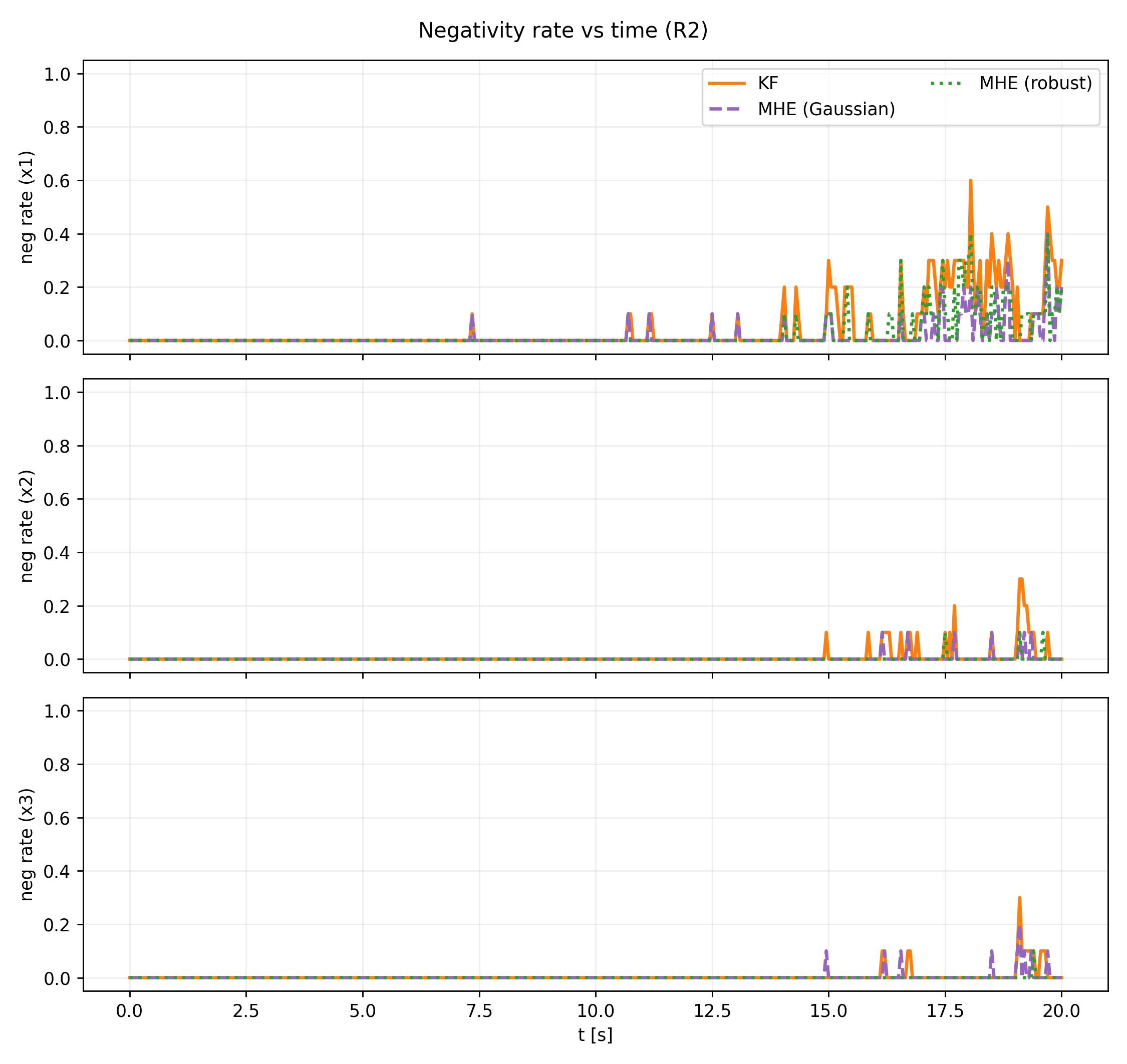}
  \caption{Regime R2 (outlier mixture)}
  \label{fig:neg_R2}
\end{subfigure}

\caption{Negativity rate of state estimates for the positive-system example under regimes R1 and R2. For each component $(x_1,x_2,x_3)$, we plot the Monte Carlo frequency of negative estimated values. This quantifies violation of the positivity prior and highlights the advantage of constrained MHE.}
\label{fig:negativity_R1_R2}
\end{figure}
\subsection*{Acknowledgment}
Work financed by AGH's subvention for scientific research.
This paper is dedicated to Professor Wojciech Mitkowski on the occasion of his 80th birthday, in recognition of his lasting contributions to control theory and to the AGH control community.
\bibliographystyle{bibtex/spmpsci}
\bibliography{references}

\end{document}